%% file: charact-prob-fuzzy.tex
\documentclass{article}
\usepackage{ijcai19}
\usepackage{times}
\usepackage{url}
\usepackage[utf8]{inputenc}
\usepackage[small]{caption}
\usepackage{amsmath}
\usepackage{amsthm}
\usepackage{booktabs}
\usepackage{enumitem}
\usepackage[author=anonymous,nomargin,marginclue,footnote,status=final]{fixme}
\FXRegisterAuthor{ls}{als}{LS}
\FXRegisterAuthor{dp}{atm}{dp}
\FXRegisterAuthor{bk}{abk}{bk}
\FXRegisterAuthor{pw}{apw}{PW}

\usepackage{amsfonts}
\usepackage{mathtools}
\usepackage[all]{xy}
\usepackage{booktabs}
\usepackage{color}

\usepackage[hidelinks]{hyperref}

\usepackage{etoolbox}
\makeatletter
\patchcmd\@combinedblfloats{\box\@outputbox}{\unvbox\@outputbox}{}{\errmessage{\noexpand patch failed}}
\makeatother

\let\OLDthebibliography\thebibliography
\renewcommand\thebibliography[1]{
  \OLDthebibliography{#1}
  \setlength{\parskip}{0pt}
  \setlength{\itemsep}{0pt plus 0.3ex}
}

\input{defsplain}

\newcommand{\ST}{\mathsf{ST}}

\newcommand{\Rat}{\mathbb{Q}}
\newcommand{\RatI}{\Rat\cap[0,1]}

\newcommand{\Games}{\mathsf{G}}

\newcommand{\Act}{\mathsf{Act}}

\newcommand{\ball}[3]{B_{#2}({#3})}

\newcommand{\id}{\mathsf{id}}

\newcommand{\Set}{\mathsf{Set}}

\newcommand{\supnorm}[1]{\lVert #1 \rVert_\infty}
\newcommand{\nbhood}[2]{U^{#1}(#2)}

\newcommand{\CF}{\mathcal{F}}
\newcommand{\qr}{\mathsf{qr}}
\newcommand{\rk}{\operatorname{\mathsf{rk}}}
\newcommand{\dif}{\,\mathrm{d}}
\newcommand{\intsuc}[2]{\textstyle{\int} #2 \dif r_{#1}}
\newcommand{\dfun}{\mathcal{D}}

\newcommand{\modf}[1]{\ALCP_{#1}}
\newcommand{\diabind}[3]{#1 \Diamond\lceil #2: #3\rceil}
\newcommand{\last}{\mathsf{last}}

\newcommand{\CI}{\mathcal{I}}
\newcommand{\CJ}{\mathcal{J}}
\newcommand{\CK}{\mathcal{K}}

\renewcommand{\theta}{\vartheta}

\newcommand{\CN}{\mathsf{N}_{\mathsf{C}}}

\newcommand{\RN}{{\mathsf{N}_{\mathsf{R}}}}
\newcommand{\ALC}{\mathcal{ALC}}

\renewcommand{\land}{\mathop{\sqcap}}
\renewcommand{\lor}{\mathop{\sqcup}}

\newcommand{\DCI}{\Delta^\CI}
\newcommand{\DCJ}{\Delta^\CJ}

\newcommand{\probably}{\textsf{P}}
\renewcommand{\Diamond}{\probably} 
\newcommand{\ALCP}{\ALC(\probably)}
\newcommand{\FOLP}{\mathsf{FO}(\probably)}

\newcommand{\ev}{\operatorname{E}}
\renewcommand{\intsuc}[2]{\ev_{r_{#1}}(#2)}
\DeclareMathOperator{\pred}{Pred}
\newcommand{\nonexpI}[1]{\pred(#1)}
\DeclareMathOperator{\cpl}{Cpl}

\newenvironment{myitemize}
{\begin{itemize}
\setlength{\itemsep}{0.1ex}
\setlength{\parsep}{0cm}
}
{\end{itemize}}

\title{A Modal Characterization Theorem for a Probabilistic Fuzzy
  Description Logic}

\author{
Paul Wild$^1$
\and
Lutz Schr\"oder$^1$\and
Dirk Pattinson$^2$\and
Barbara K\"onig$^3$
\affiliations
$^1$Friedrich-Alexander-Universit\"at Erlangen-N\"urnberg\\
$^2$Australian National University, Canberra\\
$^3$Universit\"at Duisburg-Essen
}

\begin{document}

\maketitle

\begin{abstract} The fuzzy modality \emph{probably} is interpreted
  over probabilistic type spaces by taking expected truth values. The
  arising probabilistic fuzzy description logic is invariant under
  probabilistic bisimilarity; more informatively, it is non-expansive
  wrt.\ a suitable notion of behavioural distance. In the present
  paper, we provide a characterization of the expressive power of this
  logic based on this observation: We prove a probabilistic analogue
  of the classical van Benthem theorem, which states that modal logic
  is precisely the bisimulation-invariant fragment of first-order
  logic. Specifically, we show that every formula in probabilistic
  fuzzy first-order logic that is non-expansive wrt.\ behavioural
  distance can be approximated by concepts of bounded rank in
  probabilistic fuzzy description logic.

  For a modal logic perspective on the same result,
  see~\cite{wspk:van-benthem-prob-arxiv}.
\end{abstract}

\section{Introduction}

\noindent In the representation of uncertain knowledge, one will often
wish to avoid mention of exact numerical probabilities, e.g.\ when
these are not precisely known or not relevant to the representation
task at hand -- as a typical example, a medical practitioner will
rarely name a numerical threshold for the likelihood of a diagnosis,
and instead qualify the diagnosis as, say, `suspected' or
`probable'. This has led to efforts aimed at formalizing a modality
\emph{probably}, used alternatively to modalities `with probability at
least~$p$' \cite{LarsenSkou91,HeifetzMongin01,LutzSchroder10}. Such a
formalization may be approached in a two-valued setting via
qualitative axiomatizations of likelihood
\cite{Burgess69,HalpernRabin87} or via threshold probabilities
\cite{Hamblin59,Herzig03}. In a fuzzy setting, `probably' leads a
natural life as a fuzzy modality~$\Diamond$, whose truth value just
increases as its argument becomes more probable (this modality thus
connects the otherwise well-distinguished worlds of fuzziness and
probability~\cite{LukasiewiczStraccia08}). The semantics of this
operator, first defined by Zadeh~\shortcite{Zadeh68}, interprets
$\Diamond\,\phi$ as the expected truth value of~$\phi$. It appears in
various fuzzy propositional~\cite{Hajek07,FlaminioGodo07},
modal~\cite{DesharnaisEA99,BreugelWorrell05},
fixpoint~\cite{
Kozen85,HuthKwiatkowska97}, and description
logics~\cite{SchroderPattinson11}.

In the present paper, we pin down the exact expressiveness of the
basic description logic of \emph{probably}, which we briefly refer to
as \emph{probabilistic fuzzy $\ALC$} or $\ALCP$, within a natural
ambient probabilistic fuzzy first-order logic $\FOLP$, by providing a
\emph{modal characterization theorem}. The prototype of such
characterization theorems is \emph{van Benthem's
  theorem}~\shortcite{BenthemThesis}, which states that (classical)
modal logic is precisely the bisimulation-invariant fragment of
first-order logic. It has been noted that in systems with numerical
values, \emph{behavioural pseudometrics} offer a more fine-grained
measure of equivalence than two-valued
bisimilarity~\cite{GiacaloneEA90,DesharnaisEA99,BreugelWorrell05,DesharnaisEA08,bbkk:behavioral-metrics-functor-lifting}.
When 
propositional connectives are equipped with Zadeh semantics, $\ALCP$
is \emph{non-expansive} wrt.\ behavioural distance; we continue to
refer to this property as \emph{bisimulation invariance}. In previous
work~\cite{WildEA18} we have shown that \emph{relational} fuzzy modal
logic is the bisimulation-invariant fragment of fuzzy FOL, more
precisely that every bisimulation-invariant fuzzy FO formula can be
approximated by fuzzy modal formulae \emph{of bounded rank}. The bound
on the rank is key; without it, the statement turns into a form of the
(much simpler) Hennessy-Milner theorem~\cite{HennessyMilner85} (which
classically states that non-bisimilar states in finitely branching
systems can be distinguished by modal formulae), and indeed does not
need to assume FO definability of the given bisimulation-invariant
property~\cite{BreugelWorrell05}. Here, we establish a corresponding
result for the rather more involved probabilistic setting: We show
that \emph{every bisimulation-invariant formula in probabilistic fuzzy
  FOL can be approximated in bounded rank in probabilistic fuzzy
  $\ALC$.}  This means not only that, up to approximation, $\ALCP$ is
as powerful as $\FOLP$ on bisimulation-invariant properties, but also
that $\ALCP$ provides effective syntax for bisimulation-invariant
$\FOLP$, which $\FOLP$ itself does not~\cite{Otto06}.

Proofs are mostly omitted or only sketched; 
full proofs are in the appendix.

\paragraph{Related Work} There is widespread interest in modal
characterization theorems in modal
logic~\cite{DawarOtto05}, 
database theory~\cite{FigueiraEA15}, 
concurrency~\cite{JaninWalukiewicz95,Carreiro15}, and 
AI~\cite{SturmWolter01,WildSchroder17,WildEA18}. The
overall 
structure of our proof builds partly on that of our modal
characterization theorem for relational fuzzy modal
logic~\cite{WildEA18} (in turn based ultimately on a strategy due
to Otto~\shortcite{o:van-Benthem-Rosen-elementary}) but deals with a
much more involved logic, which instead of just the lattice structure
of the unit interval involves its full arithmetic structure, via the
use of probabilities and expected values, necessitating, e.g., the use
of Kantorovich-Rubinstein duality. Notable contributions of our proof
include new forms of probabilistic bisimulation games up-to-$\epsilon$
(different from games introduced by Desharnais et
al.~\shortcite{DesharnaisEA08}, which characterize a different metric)
and Ehrenfeucht-Fra\"iss\'e games, related to two-valued games
considered in the context of topological
FOL~\cite{MakowskyZiegler80}. (For lack of space, we omit discussion
of quantitative Hennessy-Milner type results beyond the mentioned
result by van Breugel and Worrell~\shortcite{BreugelWorrell05}.)

$\FOLP$ may be seen as a fuzzy variant of
Halpern's~\shortcite{Halpern90} type-1 (i.e.\ statistical) two-valued
probabilistic FOL, and uses a syntax related to coalgebraic predicate
logic~\cite{LitakEA18} and, ultimately, Chang's \emph{modal predicate
  logic}~\cite{Chang73}. Van-Benthem style theorems for two-valued
coalgebraic modal logic~\cite{SchroderEA17} instantiate to two-valued
probabilistic modal logic, then establishing expressibility of
bisimulation-invariant probabilistic FO formulae by probabilistic
modal formulae with infinite conjunction but of bounded rank, in an
apparent analogy to bounded-rank approximation in the fuzzy setting.


\section{Fuzzy Probabilistic Logics}\label{sec:logics}

We proceed to introduce the logics featuring in our main result.  We
fix (w.l.o.g., finite) sets~$\CN$ of \emph{atomic concepts} and~$\RN$
of \emph{roles}; \emph{concepts} $C,D$ of \emph{quantitative
  probabilistic $\ALC$} ($\ALCP$) are defined by the grammar
\begin{equation*}
C,D::= q\mid A\mid C\ominus q\mid \neg C\mid C\land D
\mid \Diamond\,r.\, C
\end{equation*}
where $q\in\RatI$, $A\in\CN$ and $r\in\RN$. The intended reading
of~$\Diamond$ is `probably'; we give examples below. Slightly
deviating from standard practice, we define the \emph{rank}~$\rk(C)$
of a concept~$C$ as the maximal nesting depth of the~$\Diamond$
\emph{and atomic concepts} in~$C$; e.g.\
$\rk((\Diamond\,r.\,\Diamond\,s.\, A)\land(\Diamond\, r.\, B))=3$. We
denote the set of all concepts of rank at most~$n$ by~$\modf{n}$.

Concepts are interpreted over probabilistic structures to which we
neutrally refer as \emph{interpretations} or, briefly,
\emph{models}. We allow infinite models but restrict to discrete
probability distributions over successors at each state. A model
\begin{equation*}
\CI = (\DCI,(A^\CI)_{A\in\CN},(r^\CI)_{r\in\RN})
\end{equation*}
consists of a \emph{domain} $\DCI$ of \emph{states} or
\emph{individuals}, and interpretations $A^\CI\colon \DCI\to[0,1]$,
$r^\CI\colon \DCI\times\DCI\to[0,1]$ of atomic concepts~$A$ and
roles~$r$ such that for each $a\in\DCI$, the map
\[ r_a\colon \DCI\to[0,1], \quad r_a(a') = r^\CI(a,a') \]
is either zero or a probability mass function on $\DCI$, i.e.
\[\textstyle\sum_{a'\in\DCI}r_a(a') \in \{0,1\}\]
(implying that the \emph{support} $\{a'\in \DCI\mid r_a(a')>0\}$
of~$r_a$ is at most countable). We call a state~$a$
\emph{$r$-blocking} if $\sum_{a'\in \DCI}r_a(a') = 0$. At non-blocking
states~$a$, $r_a$ thus acts as a probabilistic accessibility relation;
we abuse~$r_a$ to denote also the probability measure defined
by~$r_a$.

The interpretation $C^\CI\colon \DCI\to[0,1]$ of concepts is defined
recursively, extending that of atomic concepts, by
\begin{align*}
q^\CI(a) &= q \\
(C\ominus q)^\CI(a) &= \max(C^\CI(a)-q,0) \\
(\neg C)^\CI(a) &= 1-C^\CI(a) \\
(C\land D)^\CI(a) & = \min(C^\CI(a),D^\CI(a))\\
(\Diamond\,r.\, C)^\CI(a) &= \ev_{r_a}(C^\CI) =
\textstyle\sum_{a'\in\DCI} r_a(a') \cdot C^\CI(a')
\end{align*}
At non-blocking~$a$, $(\Diamond\,r.\, C)^\CI(a)$ is thus the expected
truth value of~$C$ for a random $r$-successor of~$a$.  We define
disjunction~$\lor$ as the dual of~$\land$ as usual, so~$\lor$ takes
maxima. We use Zadeh semantics for the propositional operators, which
will later ensure non-expansiveness wrt.\ behavioural distance; see
additional comments in Section~\ref{sec:concl}.

Up to minor variations, our models correspond to Markov chains or, in
an epistemic reading, \emph{type spaces}
(e.g.~\cite{HeifetzMongin01}). The logic $\ALCP$ was \mbox{considered (with}
\L{}uk\-as\-iew\-icz semantics) by Schröder and
Pattinson~\shortcite{SchroderPattinson11}, and resembles van Breugel
and Worrell's quantitative probabilistic modal
logic~\shortcite{BreugelWorrell05}. E.g., in a reading of~$\Delta^\CI$
as consisting of real-world individuals, the concept
\begin{equation*}
  \mathsf{Loud}\land\Diamond\,\mathsf{hasSource}.\,(\mathsf{Large}\land
  \Diamond\,\mathsf{hasMood}.\,\mathsf{Angry})
\end{equation*}
describes noises you hear in your tent at night as being loud and
probably coming from the large and probably angry animal whose shadow
just crossed the tent roof. (In this view,~$\Diamond$ can be usefully
combined with crisp or fuzzy relational modalities, using
off-the-shelf compositionality mechanisms \cite{SchroderPattinson11}.)
In an epistemic reading where the elements of~$\Delta^\CI$ are
possible worlds, and the roles are understood as epistemic agents, the
concept
\begin{equation*}
  \neg\mathsf{GoodHand}\land\hspace{1pt} \Diamond\,\mathsf{player}.\,\Diamond\,\mathsf{opponent}.\,\mathsf{GoodHand}
\end{equation*}
denotes the degree to which $\mathsf{player}$ believes she is
successfully bluffing by letting $\mathsf{opponent}$ overestimate
$\mathsf{player}$'s hand.

For readability, we will restrict the technical treatment to a single
role~$r$, omitted in the syntax, from now on, noting that covering
multiple roles amounts to no more than additional indexing. As the
first-order correspondence language of quantitative probabilistic
$\ALC$ we introduce \emph{quantitative probabilistic first-order
  logic} ($\FOLP$), with \emph{formulae} $\phi,\psi,\dots$ defined by
the grammar
\begin{multline*}
\phi,\psi::= q\mid A(x)\mid x=y\mid\phi\ominus q\mid 
\neg\phi\mid\phi\land\psi
\mid \exists x.\,\phi\\
\mid \diabind{x}{y}{\phi}\qquad(q\in\RatI, A\in\CN)
\end{multline*}
where~$x$ and~$y$ range over a fixed countably infinite reservoir of
\emph{variables}. The reading of $\diabind{x}{y}{\phi}$ is the
expected truth value of~$\phi$ at a random successor~$y$ of~$x$.  (In
particular, when~$\phi$ is crisp, then $\diabind{x}{y}{\phi}$ is just
the probability of~$y$ satisfying~$\phi$, similar to the weights
$w_y(\phi)$ in Halpern's type-1 probabilistic
FOL~\shortcite{Halpern90}.) We have the expected notions of free and
bound variables, under the additional proviso that~$y$ (but not~$x$!)
is bound in $\diabind{x}{y}{\phi}$. The \emph{(quantifier) rank}
$\qr(\phi)$ of a formula~$\phi$ is the maximal nesting depth of the
variable-binding operators~$\exists$ and~$\Diamond$ and propositional
atoms~$A$ in~$\phi$; e.g.~$\exists x.\,\diabind{x}{y}{A(y)}$ has
rank~$3$.

Given a model $\CI = (\DCI,(A^\CI)_{A\in\CN},r^\CI)$ and a vector
$\bar a=(a_1,\dots,a_n)\in (\DCI)^n$ of values, the semantics of the
logic assigns a truth value $\phi(\bar a)\in[0,1]$ to a formula
$\phi(x_1,\dots,x_n)$ with free variables at most $x_1,\dots,x_n$. We
define $\phi(\bar a)$ recursively by essentially the same clauses as
in $\ALCP$ for the propositional constructs, and
\begin{align*}
A(x_i)(\bar a) & = A^\CI(a_i)\\
(\exists x_0.\,\phi(x_0,x_1,\dots,x_n))(\bar a) 
&  = \textstyle \bigvee_{a_0\in\DCI}\phi(a_0,a_1,\dots, a_n)\\
(\diabind{x_i}{y}{\phi(y,x_1,\dots,x_n)})(\bar a) 
& = \intsuc{a_i}{\phi(\,\cdot\,,a_1,\dots,a_n)}
\end{align*}
where~$\bigvee$ takes suprema. Moreover, equality is two-valued, i.e.\
$(x_i=x_j)(\bar a)$ is~$1$ if $a_i=a_j$, and $0$ otherwise.

E.g.\ the formula $\diabind{x}{z}{z=y}$ (`the successor of~$x$ is
probably~$y$') denotes the access probability from~$x$ to~$y$,
$\diabind{x}{z}{\diabind{z}{w}{w=y}}$ 
the probability of reaching~$y$ from~$x$ in two independently
distributed steps, and $\exists y.\,\diabind{x}{z}{z=y}$ the
probability of the most probable successor of~$x$.

We have a \emph{standard translation}~$\ST_x$ from $\ALCP$ into
$\FOLP$, indexed over a variable~$x$ naming the current
state. Following Litak et al.~\shortcite{LitakEA18}, we define
$\ST_x$ recursively by
\begin{align*}
\ST_x(A) & = A(x) \\
\ST_x(\Diamond C) & = \diabind{x}{y}{\ST_y(C)},
\end{align*}
and by commutation with all other constructs. 
\begin{lem}
  For every $\ALCP$-concept~$C$ and state~$a$,
  $C(a)=\ST_x(C)(a)$.
\end{lem}
\noindent $\ST$ thus identifies $\ALCP$
as a fragment of $\FOLP$.

\section{Behavioural Distances and Games}
\label{sec:games}

\noindent We next discuss several notions of behavioural distance
between states: via fixed point iteration \`a la
Wasserstein/Kantorovich, via games and via the logic. We focus mostly
on depth-$n$ distances. Only for one version, we define also the
unbounded distance, which will feature in the modal characterization
result. We show in Section~\ref{sec:modal-approx} that at finite
depth, all these distances coincide. It has been shown in previous
work~\cite{dgjp:metrics-labelled-markov,BreugelWorrell05} 
that the unbounded-depth distances defined via Kantorovich fixed point
iteration and via the logic, respectively, coincide in very similar
settings; such results can be seen as probabilistic variants of the
Hennessy-Milner theorem.

We recall standard notions on pseudometric spaces:

\begin{defn}[Pseudometric spaces, non-expansive maps]
  \label{def:metric}
  A \emph{(bounded) pseudometric} on a set $X$ is a function
  $d\colon X\times X\to [0,1]$ such that for $x,y,z\in X$, the
  following axioms hold: $d(x,x) = 0$ (\emph{reflexivity}),
  $d(x,y) = d(y,x)$ (\emph{symmetry}), $d(x,z) \le d(x,y)+d(y,z)$
  (\emph{triangle inequality}). If additionally $d(x,y)=0$ implies
  $x=y$, then $d$ is a \emph{metric}. A \emph{(pseudo)metric space}
  $(X,d)$ consists of a set~$X$ and a (pseudo)metric~$d$ on $X$.
  
  A map $f\colon X\to[0,1]$ is \emph{non-expansive} wrt.\ a
  pseudometric~$d$ if $|f(x)-f(y)|\le d(x,y)$ for all $x,y\in X$.  The
  space of these non-expansive functions, denoted $\nonexpI{X,d}$, is
  equipped with the \emph{supremum (pseudo)metric} $d_\infty$,
  \begin{equation*}
  d_\infty(f,g) = \supnorm{f-g} = \textstyle\bigvee_{x\in X} 
  |f(x)-g(x)|.
  \end{equation*}
  
  \noindent We denote by
  $\ball{d}{\epsilon}{x} = \{y\in X\mid d(x,y) \le \epsilon\}$ the
  \emph{ball} of radius $\epsilon$ around $x$ in $(X,d)$. The space
  $(X,d)$ is \emph{totally bounded} if for every $\epsilon > 0$ there
  exists a finite \emph{$\epsilon$-cover}, i.e.\ finitely many
  elements $x_1,\dots,x_n\in X$ such that
  $X = \bigcup_{i=1}^n \ball{d}{\epsilon}{x_i}$.

\end{defn}
\noindent Recall that a metric space is compact iff it is complete and
totally bounded.

We next introduce the Wasserstein and Kantorovich distances, which
coincide according to Kantorovich-Rubinstein duality. To this end, we
first need the notion of a coupling of two probability distributions,
from which the original distributions are factored out as marginals.

\begin{defn}
  \label{def:coupling}
  Let $\pi_1$ and $\pi_2$ be discrete probability measures on $A$ and
  $B$, respectively. We denote by $\cpl(\pi_1,\pi_2)$ the set of
  \emph{couplings} of~$\pi_1$ and~$\pi_2$, i.e.\ probability measures
  $\mu$ on $A\times B$ such that $\pi_1$ and $\pi_2$ are 
  \emph{marginals} of $\mu$:
  \begin{itemize}
    \item for all $a\in A$, $\sum_{b\in B}\mu(a,b) = \pi_1(a)$;
    \item for all $b\in B$, $\sum_{a\in A}\mu(a,b) = \pi_2(b)$.
  \end{itemize}
\end{defn}

\begin{defn}[Wasserstein and Kantorovich distances]
  \label{def:wasserstein-kantorovich}
  Let $(X,d)$ be a pseudometric space. We generally write
  \begin{equation*}
    \dfun X
  \end{equation*}
  for the set of discrete probability distributions on~$X$. We define
  two pseudometrics on~$\dfun X$, the \emph{Kantorovich
    distance}~$d^\uparrow$ and the \emph{Wasserstein
    distance}~$d^\downarrow$:
  \begin{gather*}
  d^\uparrow(\pi_1,\pi_2) =
  \textstyle\bigvee \{ |\ev_{\pi_1} (f) - \ev_{\pi_2} (f)|
    \mid f \in \nonexpI{X,d} \} \\
    d^\downarrow(\pi_1,\pi_2) =
  \textstyle\bigwedge \{ \ev_\mu (d) \mid \mu\in\cpl(\pi_1,\pi_2) \}
  \end{gather*}
  where $\bigwedge$ takes meets (and~$\bigvee$ suprema).  We extend
  these distances without further mention to zero functions (like the
  functions~$r_a$ at blocking states~$a$) by decreeing that the zero
  function has distance~$1$ from all probability distributions.
\end{defn}

\noindent The notation $d^\uparrow,d^\downarrow$ is meant as a
mnemonic for the fact that these distances are obtained via suprema
respectively via infima.  If $(X,d)$ is separable (contains a
countable dense subset), these pseudometrics coincide, a fact known as
the \emph{Kantorovich-Rubinstein duality} (e.g.~\cite{dudley2002}):

\begin{lem}[Kantorovich-Rubinstein duality] \label{lem:kr-duality}
Let $(X,d)$ be a separable pseudometric space. Then for all 
$\pi_1,\pi_2\in\dfun X$,
\[ d^\uparrow(\pi_1,\pi_2) = d^\downarrow(\pi_1,\pi_2). \]
\end{lem}

\noindent The above notions of \emph{lifting} a distance on~$X$ to a
distance on distributions over~$X$ can be used to give fixed point
equations for behavioural distances on models.

\begin{defn}[Fixed point iteration \`a la Wasserstein/Kantorovich]
  \label{def:fixed-point-iteration}
  Given a model $\CI$, we define the chains $(d^K_n)$, $(d^W_n)$ of
  \emph{depth-$n$ Kantorovich} and \emph{Wasserstein distances},
  respectively, via fixed point iteration:
  \begin{gather*}
    d^W_0(a,b) = d^K_0(a,b) = 0 \\
    d^W_{n+1}(a,b) = \textstyle\bigvee_{A\in\CN}|A^\CI(a)-A^\CI(b)| \vee
    (d^W_n)^\downarrow(\pi_a,\pi_b) \\
    d^K_{n+1}(a,b) = \textstyle\bigvee_{A\in\CN}|A^\CI(a)-A^\CI(b)| \vee
    (d^K_n)^\uparrow(\pi_a,\pi_b)
  \end{gather*}
  where $\vee$ is binary join.  We extend this to states $a,b$ in
  different models~$\CI$, $\CJ$ by taking the disjoint union of~$\CI$,
  $\CJ$.
\end{defn}

\noindent In both cases, we start with the zero pseudometric, and in
the next iteration lift the pseudometric~$d_n$ from the previous step
via Wasserstein/Kantorovich. This lifted metric is then applied to the
probability distributions $\pi_a,\pi_b$ associated with $a,b$. In
addition we take the maximum with the supremum over the distances for
all atomic $A\in\CN$.

We now introduce a key tool for our technical development, a new
up-to-$\epsilon$ bisimulation game inspired by the definition of the
Wasserstein distance.

\begin{defn}[Bisimulation game]
  \label{def:bisimulation-game}
  Given models $\CI,\CJ$, $a_0\in \DCI,b_0\in \DCJ$, and
  $\epsilon_0\in[0,1]$, the \emph{$\epsilon_0$-bisimulation game} for
  $a_0$ and $b_0$ is played by \emph{Spoiler} ($S$) and
  \emph{Duplicator} ($D$), with rules as follows:
  \begin{myitemize}
  \item \emph{Configurations}: triples $(a,b,\epsilon)$, with states
    $a\in \DCI$, $b\in \DCJ$ and maximal allowed deviation
    $\epsilon\in[0,1]$. The \emph{initial configuration} is
    $(a_0,b_0,\epsilon_0)$.
    \item \emph{Moves}: In each round, $D$ first picks a probability measure
    $\mu \in \cpl(\pi_a,\pi_b)$.
    Then, $D$ distributes the deviation~$\epsilon$ over all pairs
    $(a',b')$ of successors, i.e.~picks a function $\epsilon'\colon
    \DCI\times\DCJ\to[0,1]$ such that $\ev_\mu(\epsilon')\le\epsilon$.
    Finally, $S$ picks a pair $(a',b')$ with $\mu(a',b') > 0$; the
    new configuration is then $(a',b',\epsilon'(a',b'))$.
    \item $D$ \emph{wins} if both states are blocking or $\epsilon = 1$.
    \item $S$ \emph{wins} if exactly one state is blocking and $\epsilon <
    1$.
  \item \emph{Winning condition}:
    $|A^\CI(a)-A^\CJ(b)|\hspace{-0.4mm}\le\hspace{-0.4mm}\epsilon$ for
    all $A\in\CN$.
  \end{myitemize}
  \noindent The game comes in two variants, the \emph{(unbounded)
    bisimulation game} and the \emph{$n$-round bisimulation game},
  where $n\ge 0$. Player $D$ wins if the winning condition holds
  \emph{before} every round, otherwise $S$ wins. More precisely, $D$
  wins the unbounded game if she can force an infinite play and the
  $n$-round game once~$n$ rounds have been played (the winning
  condition is not checked after the last round, so in particular, any
  $0$-round game is an immediate win for $D$).
\end{defn}

\begin{rem}
  The above bisimulation game differs from bisimulation games in the
  literature (e.g.~\cite{DesharnaisEA08}) in a number of salient
  features. A particularly striking aspect is that~$D$'s moves are not
  similar to those of~$S$, and moreover~$D$ in fact moves
  before~$S$. Intuitively,~$D$ is required to commit beforehand to a
  strategy that she will use to respond to~$S$'s next move. Note also
  that the precision~$\epsilon$ changes as the game is being played, a
  complication forced by the arithmetic nature of models.
\end{rem}

\noindent This leads to notions of game distance:

\begin{defn}
  \label{def:game-distance}
  \emph{depth-$n$ game distance}~$d^G_n$ and
  (unbounded-depth) \emph{game distance}~$d^G$ are defined as
  \begin{align*}
    d^G_n(a,b)& = \textstyle\bigwedge\{ \epsilon\mid D\text{ wins }\Games_n(a,b,\epsilon)\}\\
    d^G(a,b) & = \textstyle\bigwedge\{ \epsilon\mid D\text{ wins }\Games(a,b,\epsilon)\}. 
  \end{align*}
  where $\Games(a,b,\epsilon)$ and $\Games_n(a,b,\epsilon)$ denote the
  the bisimulation game and the $n$-round bisimulation game on
  $(a,b,\epsilon)$, respectively.
\end{defn}

\noindent Finally we define the depth-$n$ logical distance via
$\ALCP$, restricting to concepts of rank at most~$n$:

\begin{defn}
  \label{def:logical-distance}
  The \emph{depth-$n$ logical distance}~$d^L_n(a,b)$ of states $a$, $b$
  in models $\CI$, $\CJ$ is defined as
  \begin{equation*}
  d^L_n(a,b) = \textstyle\bigvee \{ |C^\CI(a)-C^\CJ(b)| \mid \rk(C)\le n \}.
  \end{equation*}
\end{defn}
\noindent The equivalence of the four bounded-depth behavioural
distances introduced above will be shown in
Theorem~\ref{thm:modal-approx}.

\bigskip

\noindent Behavioural distance forms the yardstick for our notion of
bisimulation invariance; for definiteness:
\begin{defn}
  A quantitative, i.e.\ $[0,1]$-valued, property~$Q$ of states, or a
  formula or concept defining such a property, is
  \emph{bisimulation-invariant} if~$Q$ is non-expansive wrt.\ game
  distance, i.e.\ for states $a,b$ in models $\CI,\CJ$, respectively,
  \begin{equation*}
  |Q(a)-Q(b)|\le d^G(a,b).
  \end{equation*}
  Similarly,~$Q$ is \emph{depth-$n$ bisimulation invariant}, or
  \emph{finite-depth bisimulation invariant} if mention of~$n$ is
  omitted, if~$Q$ is non-expansive wrt.\ $d^G_n$ in the same
  sense. 
\end{defn}
\noindent It is easy to see that \emph{$\ALCP$-concepts are
  bisimulation-invariant}. More precisely, $\ALCP$-concepts of rank
at most~$n$ are depth-$n$ bisimulation invariant (a stronger
invariance since clearly $d^G_n\le d^G$), as shown by routine
induction.
In contrast, many other properties of states are expressible in
$\FOLP$ but not in $\ALCP$, as they fail to be
bisimulation-invariant. Examples include $\diabind{x}{y}{x=y}$
(probability of a self-transition) and $\exists z.\,
\diabind{x}{y}{y=z}$ (highest transition probability to a successor).

We are now ready to formally state our main theorem (a proof will be
given in Section~\ref{sec:main}):
\begin{thm}[Modal characterization]\label{thm:van-benthem}
  Every bisimulation-invariant $\FOLP$-formula of rank at most $n$ can
  be approximated (uniformly across all models) by
  $\ALCP$-concepts of rank at most $3^n$.
\end{thm}
\noindent (The exponential bound on the rank features also in the full
statement of van Benthem's theorem.)

\section{Modal Approximation at Finite Depth}
\label{sec:modal-approx}

We now establish the most important stepping stone on the way to the
eventual proof of the modal characterization theorem: We show that
every depth-$n$ bisimulation-invariant property of states can be
approximated by $\ALCP$-concepts of rank at most~$n$. We prove this
simultaneously with coincidence of the various finite-depth
behavioural pseudometrics defined in the previous section. To begin,

\begin{lem}\label{lem:metrics-equal-GW}
  The game-based pseudometric~$d^G_n$ coincides with the Wasserstein
  pseudometric~$d^W_n$,
\end{lem}
\noindent 
We note next that the modality~$\Diamond$ is non-expansive: We
extend~$\Diamond$ to act on $[0,1]$-valued functions
$f\colon\DCI\to[0,1]$ by
\begin{align*}
(\Diamond f)(a)&=\intsuc{a}{f}.
\end{align*}
\begin{lem}\label{lem:diamond-nonexp}
  The map $f \mapsto \Diamond f$ is non-expansive wrt.\ the supremum
  metric, that is $\supnorm{\Diamond f - \Diamond g} \le 
  \supnorm{f-g}$ for all $f,g\colon \DCI\to[0,1]$.
  \pwnote{Find a better spot for this lemma?}
\end{lem}

\noindent Following our previous work~\cite{WildEA18}, we prove
coincidence of the remaining pseudometrics in one big induction, along
with total boundedness (needed later to apply a variant of the
Arzel\`a-Ascoli theorem and the Kantorovich-Rubinstein duality) and
modal approximability of depth-$n$ bisimulation-invariant
properties. We phrase the latter as density of the (semantics of)
$\ALCP$-concepts of rank at most~$n$ in the non-expansive function
space (Definition~\ref{def:metric}):

\begin{thm}\label{thm:modal-approx}
  Let $\CI$ be a model. Then for all $n\ge 0$,
  \begin{enumerate}
    \item we have $d^G_n = d^W_n = d^K_n = d^L_n =: d_n$ on $\CI$;
    \label{item:metrics-equal}
    \item the pseudometric space $(\DCI,d_n)$ is 
    totally 
    bounded;
    \label{item:tot-bounded}
    \item $\modf{n}$ is a dense subset of $\nonexpI{\DCI,d_n}$.
    \label{item:modal-approx}
  \end{enumerate}
\end{thm}

\begin{proof}[Proof 
  sketch]
  By simultaneous induction on~$n$.
  
  In the base case $n = 0$, all the behavioural distances are the 
  zero pseudometric, so that total boundedness follows trivially and 
  the density claim follows because non-expansive maps are just 
  constants in $[0,1]$ and the syntax of $\ALCP$ includes truth 
  constants $q\in\RatI$.
  
  For the inductive step, let $\CI$ be a model and $n > 0$, and assume
  as the inductive hypothesis that all claims in
  Theorem~\ref{thm:modal-approx} hold for all $n' < n$.  We begin with
  Item~\ref{item:metrics-equal}; $d^G_n = d^W_n$ is
  already proved (Lemma~\ref{lem:metrics-equal-GW}).
  \begin{itemize}[wide]
  \item $d^W_n = d^K_n$ follows by Kantorovich-Rubinstein duality
    (Lemma~\ref{lem:kr-duality}), since every totally bounded
    pseudometric space is separable.
  \item $d^K_n = d^L_n$: By Lemma~\ref{lem:diamond-nonexp} and the
    inductive hypothesis, $\Diamond[\modf{n-1}]$ is dense in
    $\Diamond[\nonexpI{\DCI,d_{n-1}}]$. Thus, the supremum in the
    definition of $d^K_n$ does not change when it is taken only over
    the concepts in $\modf{n-1}$ instead of all nonexpansive
    properties. The proof is finished by a simple induction over
    propositional combinations of concepts.
  \end{itemize}
  
  \noindent \emph{Item~\ref{item:tot-bounded}}: By the inductive
  hypothesis, the space $(\DCI,d_{n-1})$ is totally bounded. By the
  Arzel\`a-Ascoli theorem (in a version for totally bounded spaces and
  non-expansive maps, cf.~\cite{WildEA18}), it follows that
  $\nonexpI{\DCI,d_{n-1}}$ is totally bounded wrt.~the supremum
  pseudometric.  This implies that depth-$n$ distances can be
  approximated up to $\epsilon$ by examining differences at only
  finitely many, say $m$, concepts. As $([0,1]^m,d_\infty)$ is totally
  bounded, $(\DCI,d_n)$ is, too.
  
  \emph{Item~\ref{item:modal-approx}}: By the Stone-Weierstra\ss{}
  theorem (again in a version for totally bounded spaces and
  non-expansive maps~\cite{WildEA18}) it suffices to give, for each
  $\epsilon>0$, each non-expansive map $f\in\nonexpI{\DCI,d_n}$, and
  each pair of states $a,b\in\DCI$ a concept $C\in\modf{n}$ such that
  \begin{equation*}
    \max(|f(a)-C^\CI(a)|, |f(b)-C^\CI(b)|) \le \epsilon.
  \end{equation*}
  To construct such a~$C$, we note that $|f(a)-f(b)|\le d^L_n(a,b)$
  (by non-expansiveness), so there exists some $D\in\modf{n}$ such
  that $|D^\CI(a)-D^\CI(b)|\ge |f(a)-f(b)|-\epsilon$. From~$D$, we can
  construct~$C$ using truncated subtraction~$\ominus$.
\end{proof}

\noindent This completes the proof of
Theorem~\ref{thm:modal-approx}. Now that we can approximate depth-$k$
bisimulation-invariant properties by $\ALCP$-concepts of rank~$k$ on
any fixed model, we need to make the approximation uniform across all
models. We achieve this by means of a \emph{final} model, i.e.\ one
that realizes all behaviours. 
Formally:
\begin{defn}
  A \emph{(probabilistic) bounded morphism} between models~$\CI$,
  $\CJ$ is a map $f:\DCI\to\DCJ$ such that $A^\CI=f^{-1}[A^\CJ]$ for
  each $A\in\CN$ and $r_{f(a)}(B)=r_a(f^{-1}[B])$ for all
  $B\subseteq\DCJ$, $a\in\DCI$ (implying that~$a$ is blocking iff
  $f(a)$ is blocking). A model~$\CF$ is \emph{final} if for every
  model~$\CI$, there exists a unique bounded morphism $\CI\to\CF$.
\end{defn}
\noindent It follows from standard results in coalgebra~\cite{Barr93}
that a final model exists. Bounded morphisms preserve behaviour
on-the-nose, that is:
\begin{lem}\label{lem:bounded-morphism-game}
  Let $f\colon\CI\to\CJ$ be a bounded morphism. Then, for any
  $a\in\Delta^\CI$, $d^G(a,f(a)) = 0$.
\end{lem}

\noindent This entails the following lemma, which will enable us to
use approximants on the final model as uniform approximants across all
models:
\begin{lem}\label{lem:uniform-approx}
  Let~$\CF$ be a final model, and let $\phi$ and $\psi$ be
  bisimulation-invariant first-order properties. Then, for any
  model~$\CI$, $\supnorm{\phi-\psi}^\CI \le \supnorm{\phi-\psi}^\CF$.
\end{lem}

\section{Locality}\label{sec:locality}

The proof of the modal characterization theorem now further proceeds
by first establishing that every bisimulation-invariant first-order
formula~$\phi$ is \emph{local} in a sense to be made precise shortly,
and subsequently that~$\phi$ is in fact even finite-depth bisimulation
invariant, for a depth that is exponential in the rank
of~$\phi$. Locality refers to a probabilistic variant of Gaifman
graphs~\cite{Gaifman82}:

\begin{defn}
  Let $\CI$ be a model.
  \begin{itemize}[wide]
    \item The \emph{Gaifman graph} of $\CI$ is the undirected graph on
    the set~$\DCI$ of vertices that has an edge for every pair 
    $(a,b)$ with $r^\CI(a,b) > 0$ or $r^\CI(b,a) > 0$.
  \item The \emph{Gaifman distance}
    $D \colon \DCI\times\DCI \to \mathbb{N}\cup\{\infty\}$ is graph
    distance in the Gaifman graph: For every $a,b\in\DCI$, the
    distance $D(a,b)$ is the least number of edges on any path from~$a$
    to $b$, if such a path exists, and $\infty$ otherwise.
  \item For $a\in\DCI$ and $k \ge 0$, the \emph{radius $k$
      neighbourhood}
    $\nbhood{k}{a} = \{b \in \DCI \mid D(a,b) \le k \}$ of~$a$
    consists of the states reachable from~$a$ in at most $k$ steps.  
  \item The \emph{restriction} of $\CI$ to $\nbhood{k}{a}$ is the
    model $\CI^k_a$ with set $\nbhood{k}{a}$ of states, and
    \begin{align*}
      A^{\CI^k_a}(b) & = A^\CI(b) &
      r^{\CI^k_a}(b,c) & =
                         \begin{cases}
                           r^\CI(b,c) & \text{if } D(a,b) < k \\
                           0 & \text{if } D(a,b) = k \\
                         \end{cases}
    \end{align*}
    for $A\in\CN$ and $b,c\in\nbhood{k}{a}$.
  \end{itemize}
\end{defn}
\noindent The restriction to $\nbhood{k}{a}$ thus makes all states at
distance~$k$ blocking. 
Restricted models have the expected relationship with games
of bounded depth:

\begin{lem}\label{lem:nbhood-bisim}
  Let~$a$ be a state in a model~$\CI$. Then~$D$ wins the $k$-round
  $0$-bisimulation game for $\CI,a$ and $\CI^k_a,a$.
\end{lem}

\noindent Locality of a formula now means that its truth values only
depend on the neighbourhood of the state in question:

\begin{defn}
  A formula $\phi(x)$ is \emph{$k$-local} for a radius~$k$ if for
  every model $\CI$ and every $a\in \DCI$,
  $\phi^\CI(a) = \phi^{\CI^k_a}(a)$.
\end{defn}
\noindent As $\ALCP$-concepts are bisimulation-invariant,
Lemma~\ref{lem:nbhood-bisim} implies

\begin{lem}
  Every $\ALCP$-concept of rank at most $k$ is $k$-local.
\end{lem}
\noindent To prove locality of bisimulation-invariant
$\FOLP$-formulae, we require a model-theoretic tool, an adaptation of
Eh\-ren\-feucht-Fraïss{\'e} equivalence to the probabilistic setting:

\begin{defn}
  \label{def:ef-game}
  Let $\CI,\CJ$ be models, and let~$\bar a_0$ and~$\bar b_0$
  be vectors of equal length over~$\DCI$ and~$\DCJ$, respectively.
  The \emph{Ehrenfeucht-Fra\"iss\'e game for $\CI,\bar a_0$ and
    $\CJ,\bar b_0$}, played by \emph{Spoiler} ($S$) and
  \emph{Duplicator} ($D$), is given as follows.
  \begin{itemize}
  \item\emph{Configurations:} pairs $(\bar a,\bar b)$ of vectors
    $\bar a$ over~$\DCI$ and~$\bar b$ over~$\DCJ$; the \emph{initial
      configuration} is $(\bar a_0,\bar b_0)$.
    \item \emph{Moves:}
    Each round can be played in one of two ways, chosen by $S$:
    \begin{itemize}[wide]
    \item \emph{Standard round}: $S$ selects a state in one model, say
      $a\in\DCI$, and $D$ then has to select a state in the other
      model, say $b\in\DCJ$, reaching the configuration
      $(\bar aa,\bar bb)$.
    \item \emph{Probabilistic round}: $S$ selects an index~$i$ and a
      fuzzy subset in one model, say $\phi_A\colon\DCI\to [0,1]$.  $D$
      then has to select a fuzzy subset in the other model, say
      $\phi_B\colon\DCJ\to [0,1]$, such that
      $\intsuc{a_i}{\phi_A} = \intsuc{b_i}{\phi_B}$. Then, $S$ selects
      an element on one side, say $a\in\DCI$, such that 
      $r_{a_i}(a)>0$,
      and $D$ subsequently selects an element on the other side, say
      $b\in\DCJ$, such that $\phi_A(a) = \phi_B(b)$ and 
      $r_{b_i}(b)>0$,
      reaching the configuration $(\bar aa,\bar bb)$.
    \end{itemize} 
  \item \emph{Winning conditions}: Any player who cannot move
    loses. $S$ wins if a configuration is reached (including the
    initial configuration) that fails to be a partial
    isomorphism. Here, a configuration $(\bar a,\bar b)$ is a
    \emph{partial isomorphism} if
    \begin{itemize}
      \item $a_i=a_j\iff b_i=b_j$
      \item $A^\CI(a_i) = A^\CJ(b_i)$ for all $i$ and all $A\in\CN$
      \item $r^\CI(a_i,a_j) = r^\CJ(b_i,b_j)$ for all $i,j$.
    \end{itemize}
    Player~$D$ wins if she reaches the $n$-th round (maintaining
    configurations that are not winning for $S$).
  \end{itemize}
\end{defn}
\noindent
For our purposes, we need only soundness of
Ehrenfeucht-Fra\"iss\'e equivalence:
\begin{lem}[Ehrenfeucht-Fra\"iss\'e invariance]
  \label{lem:ef-inv-fol}
  Let $\CI,\CJ$ be models, and let $\bar a_0,\bar b_0$ be vectors of
  length $m$ over~$\DCI$ and $\DCJ$, respectively, such that $D$ wins
  the $n$-round Ehrenfeucht-Fra\"iss\'e game on $\bar a_0,\bar b_0$.
  Then for every $\FOLP$-formula $\phi$ with $\qr(\phi)\le n$ and free
  variables at most $x_1,\dots,x_m$,
  \begin{equation*}
  \phi(\bar a_0) = \phi(\bar b_0).
  \end{equation*}  
\end{lem}

\noindent Since embeddings into disjoint unions of models are
bounded morphisms, the following is immediate from
Lemma~\ref{lem:bounded-morphism-game}:
\begin{lem}\label{lem:bisim-inv-disjoint}
  Every bisimulation-invariant formula is also invariant under
  disjoint union.
\end{lem}
\noindent We are now in a position to prove our desired locality
result:
\begin{lem}[Locality]\label{lem:bisim-inv-local}
  Let $\phi(x)$ be a bisimulation-invariant $\FOLP$-formula of rank
  $n$ with one free variable~$x$. Then $\phi$ is $k$-local for
  $k = 3^n$.
\end{lem}
\begin{proof}[Proof sketch]
  Let $a$ be a state in a model~$\CI$. We need to show
  $\phi^\CI(a) = \phi^{\CI^k_a}(a)$.  Construct models~ $\CJ,\CK$ that
  extend $\CI$ and~$\CI^k_a$, respectively, by adding $n$ disjoint
  copies of both $\CI$ and~$\CI^k_a$. We finish the proof by showing
  that
  \begin{equation*}
  \phi^\CI(a) = \phi^\CJ(a) = \phi^\CK(a) = \phi^{\CI^k_a}(a).
  \end{equation*}
  The first and third equality follow by bisimulation invariance
  of~$\phi$ (Lemma~\ref{lem:bisim-inv-disjoint}), and the second using
  Lemma~\ref{lem:ef-inv-fol}, by giving a winning invariant for~$D$ in
  the $n$-round Ehrenfeucht-Fra\"iss\'e game for $\CJ,a$ and
  $\CK,a$.  
\end{proof}

\section{Proof of the Main Result}
\label{sec:main}
\noindent Having established locality of bisimulation-invariant
first-order formulae and modal approximability of finite-depth
bisimulation-invariant properties, we now discharge the last remaining
steps in our programme: We show by means of an unravelling
construction that bisimulation-invariant first-order formulae are
already finite-depth bisimulation-invariant, and then conclude the
proof of our main result, the modal characterization theorem.
\begin{defn}
  Let $\CI$ be a model. The \emph{unravelling} $\CI^\ast$ of $\CI$ is
  a model with non-empty finite sequences $\bar a\in (\DCI)^+$ as
  states, where atomic concepts and roles are interpreted by
  \begin{equation*}
  A^{\CI^\ast}(\bar a) = A^\CI(\last(\bar a)) \qquad
  r^{\CI^\ast}(\bar a,\bar aa) = r^\CI(\last(\bar a),a),
  \end{equation*}
  for $\bar a \in (\DCI)^+$ and $a\in\DCI$, where $\last$ takes last
  elements.
\end{defn}

\noindent As usual, models are bisimilar to their unravellings:

\begin{lem}\label{lem:bisim-unravel}
  For any model $\CI$ and $a\in\DCI$, $D$ has
  a winning strategy in the $0$-bisimulation game for $\CI,a$ and 
  $\CI^\ast,a$.
\end{lem}

\noindent We next show that locality and bisimulation invariance imply
finite-depth bisimulation invariance:
\begin{lem}\label{lem:local-k-bisim-inv}
  Let $\phi$ be bisimulation invariant and $k$-local.
  Then $\phi$ is depth-$k$ bisimulation invariant.
\end{lem}
\begin{proof}[Proof sketch]
  By unravelling (Lemma~\ref{lem:bisim-unravel}) and loc\-ality
  (Lemma~\ref{lem:nbhood-bisim}), we need only consider depth-$k$ tree
  models. On such models, winning strategies in $k$-round bisimulation
  games automatically win also the unrestricted game.
\end{proof}
\noindent This allows us to wrap up the proof of our main result:
\begin{proof}[Proof of Theorem~\ref{thm:van-benthem}]
  Let $\phi$ be a probabilistic first-order formula of rank $n$. By
  Lemma~\ref{lem:bisim-inv-local} and
  Lemma~\ref{lem:local-k-bisim-inv}, $\phi$ is depth-$k$
  bisimulation-invariant for $k = 3^n$.  By
  Theorem~\ref{thm:modal-approx}, for every $\epsilon>0$, there exists
  an $\ALCP$ concept $C_\epsilon$ of rank at most $k$ such that
  $\supnorm{\phi^\mathcal{F}-C^\mathcal{F}_\epsilon}\le\epsilon$ on
  the final model $\mathcal{F}$. By Lemma~\ref{lem:uniform-approx},
  this approximation works over all models.
\end{proof}

\section{Conclusions}\label{sec:concl}

\noindent We have established a modal characterization result for a
probabilistic fuzzy DL~$\ALCP$, stating that every formula of
quantitative probabilistic FOL that is \emph{bisimulation-invariant},
i.e.\ non-expansive wrt.\ a natural notion of behavioural distance,
can be approximated by $\ALCP$-concepts of bounded modal rank, the
bound being exponential in the rank of the original formula. As
discussed in the introduction, the bound on the modal rank is the
crucial feature making this result into a van-Benthem (rather than
Hennessy-Milner) type theorem.

It remains open whether our main result can be sharpened to make do
without approximation. (Similar open problems persist for the case of
fuzzy modal logic~\cite{WildEA18} and two-valued probabilistic modal
logic~\cite{SchroderEA17}.) Further directions for future research
include 
a treatment of \L{}ukasiewicz semantics of the propositional
connectives (for which non-expansiveness in fact fails). Moreover, the
version of our main result that restricts the semantics to finite
models, in analogy to Rosen's finite-model version of van Benthem's
theorem~\cite{Rosen97}, remains open.

\newpage
\bibliographystyle{mynamed}
\interlinepenalty=10000
\bibliography{coalgml}

\newpage
\appendix
\section{Appendix}

\subsection{Coalgebraic Modelling}\label{sec:coalg}
\lsnote{May still need adaptation} Universal coalgebra~\cite{Rutten00}
serves as a generic framework for modelling state-based systems, with
the system type encapsulated as a set functor. Although we are only
concerned with a concrete system type in the present paper, we do need
coalgebraic methods to some degree. In particular, the requisite
background on behavioural
distances~\cite{BreugelWorrell05,bbkk:behavioral-metrics-functor-lifting}
is largely based on coalgebraic techniques, and moreover we will need
the final coalgebra at one point in the development. We require only
basic definitions, which we recapitulate here and then instantiate to
the case of our notion of model.

Recall first that a set functor~$F:\Set\to\Set$ consists of an
assignment of a set~$FX$ to every set~$X$ and a map $Ff:FX\to FY$ to
every map $f:X\to Y$, preserving identities and composition. The core
example of a functor for the present purposes is the
\emph{distribution functor}~$\dfun$, which assigns to a set $X$ the
set $\dfun X$ of discrete probability measures on~$X$, and to a map
$f:X\to Y$ the map $\dfun f:\dfun X\to\dfun Y$ that takes image
measures; explicitly, $\dfun f(\mu)$ is the image measure of~$\mu$
along~$f$, given by $\dfun f(\mu)(A)=\mu(f^{-1}[A])$. Functors can be
combined by taking \emph{products} and \emph{sums}: Given set functors
$F,G:\Set\to\Set$, the set functors $F\times G,F+G:\Set\to\Set$ are
given by $(F\times G)X=FX\times GX$ and $(F+G)X=FX+GX$, respectively,
with the evident action on maps in both cases; here, $+$ denotes
disjoint union as usual. Every set~$C$ induces a \emph{constant
  functor}, also denoted~$C$ and given by $CX=C$ and $Cf=\id_C$ for
every set~$X$ and every map~$f$. Moreover, the \emph{identity
  functor}~$\id$ is given by $\id\, X=X$ and $\id\, f=f$ for all 
  sets~$X$
and all maps~$f$.

An \emph{$F$-coalgebra} $(A,\xi)$ for a set functor~$F$ consists of a
set~$X$ of \emph{states} and a \emph{transition map} $\xi:A\to FA$,
thought of as assigning to each state $a\in A$ a structured collection
$\xi(a)$ of successors.  A $\dfun$-coalgebra $(A,\xi)$, for instance,
is just a Markov chain: its transition map $\xi:A\to\dfun A$ assigns
to each state a distribution over successor states. Similarly, models
in the sense defined above are coalgebras $(A,\xi)$ for the set
functor $[0,1]^{\CN}\times(\dfun+1)$: If $\xi(a)=(f,\pi)$, then
$f:\CN\to[0,1]$ determines the truth values of the atomic concepts
at the state~$a$, and $\pi$ is either a discrete probability measure
determining the successors of~$a$ or a designated value denoting
termination. The probabilistic transition systems considered by van
Breugel and Worrell \cite{BreugelWorrell05}, which indexes
probabilistic transition relations over a set~$\Act$ of actions and
moreover uses unrestricted subdistributions, corresponds to coalgebras
$(A,\xi)$ for the set functor $\dfun(\id+1)^\Act$ -- given a state~$a$
and an action $c\in\Act$, $\xi(a)(c)\in\dfun(A+1)$ is a
subdistribution over successor states of~$a$, with the summand~$1$
serving to absorb the weight missing to obtain total
weight~$1$.\bknote{Actually another difference is that van
  Breugel/Worrell work in a non-discrete setting in measurable
  spaces.}

A \emph{morphism} $f:(A,\xi)\to(B,\zeta)$ between $F$-coalgebras
$(A,\xi)$ and $(B,\zeta)$ is a map $f:A\to B$ such that
\begin{equation*}
Ff(\xi(a)) = \zeta(f(a)) 
\end{equation*}
for all states $a\in A$. Morphisms should be thought of as
behaviour-preserving maps or functional bisimulations. E.g.\
$f:A\to B$ is a morphism of $\dfun$-coalgebras (i.e.\ Markov chains)
$(A,\xi)$ and $(B,\zeta)$ if for each set $Y\subseteq B$ and each
state $a\in A$,
\begin{equation*}
\zeta(f(a))(Y)=\xi(a)(f^{-1}[Y]),
\end{equation*}
i.e.\ the probability of reaching $Y$ from $f(a)$ is the same as that
of reaching $f^{-1}[Y]$ from~$a$. Morphisms of probabilistic
transition systems, viewed as coalgebras, satisfy a similar condition
for the successor distributions, and additionally preserve the truth
values of atomic concepts.

An $F$-coalgebra $(Z,\zeta)$ is \emph{final} if for every
$F$-coalgebra $(A,\xi)$ there exists exactly one morphism
$(A,\xi)\to(Z,\zeta)$. Final coalgebras are unique up to isomorphism
if they exist, and should be thought of as having as states all
possible behaviours of states in $F$-coalgebras. For our present
purposes, we do not need an explicit description of the final
coalgebra; it suffices to know that since the functor describing
probabilistic transition systems is \emph{accessible} (more precisely
$\omega_1$-accessible), a final coalgebra for it, i.e.\ a final
probabilistic transition system, exists~\cite{Barr93}. 

\subsection{Omitted Proofs}

\subsubsection{Proof of Lemma~\ref{lem:kr-duality}}

We make use of the following version of the Kantorovich-Rubinstein
duality~\cite[Proposition 11.8.1]{dudley2002}:
\begin{lem}[Kantorovich-Rubinstein duality] 
\label{lem:kr-duality-cited} Let
  $(X,d)$ be a separable metric space, and let $\mathcal{P}_1(X)$
  denote the space of probability measures
  $\mu\colon\mathcal{B}(X) \to [0,1]$ on the Borel $\sigma$-algebra
  $\mathcal{B}(X)$ such that
  $\textstyle{\int} d(x,\,\cdot\,) \dif\mu < \infty$ for some $x\in
  X$. Then for $\mu_1,\mu_2\in\mathcal{P}_1(X)$,
  \begin{equation*}
  d^\uparrow(\mu_1,\mu_2) = d^\downarrow(\mu_1,\mu_2).
  \end{equation*}
\end{lem}

Essentially, we only need to transfer this version of 
Kantorovich-Rubinstein duality to the slightly more general case of 
pseudometrics.

First, note that the relation $x\sim y :\iff d(x,y) = 0$ is an 
equivalence relation on $X$. The quotient set $Y := X/{\sim}$ is made 
into a metric space $(Y,d')$, the \emph{metric quotient} of $(X,d)$, 
by taking $d'([x],[y])=d(x,y)$. Let $p\colon A \to B$ be the 
projection map. By construction, $p$ is an
isometry. Both the Kantorovich and the Wasserstein lifting preserve
isometries~\cite{bbkk:behavioral-metrics-functor-lifting}, so for
all discrete probability measures $\mu_1,\mu_2$ on $X$,
\begin{align*}
d^{\uparrow}(\mu_1,\mu_2) 
& = (d')^{\uparrow}((\dfun p)\mu_1,(\dfun p)\mu_2) \\
& = (d')^{\downarrow}((\dfun p)\mu_1,(\dfun p)\mu_2) \\
& = d^{\downarrow}(\mu_1,\mu_2).
\end{align*}
In the second step we have applied Lemma~\ref{lem:kr-duality-cited} 
to the metric space $(Y,d')$, noting that every discrete probability 
measure can be defined on the Borel $\sigma$-algebra.

\subsubsection{Proof of Lemma~\ref{lem:metrics-equal-GW}.}

Induction over~$n$. The base case $n = 0$ is clear: the $0$-round
game is an immediate win for~$D$, so $d^G_0 = d^W_0 = 0$. We proceed
with the inductive step from~$n$ to~$n+1$.

So let~$a$ and~$b$ be states in a model $\CI$. If $a$ and $b$ are both
blocking, then $d^G_{n+1}(a,b) = d^W_{n+1}(a,b) = 0$. If exactly
one of $a,b$ is blocking, then
$d^G_{n+1}(a,b) = d^W_{n+1}(a,b) = 1$. Now assume that both $a$ and
$b$ are non-blocking.

``$\ge$'': Let $d^G_{n+1}(a,b) \le \epsilon$, so $D$ wins the
$(n+1)$-round bisimulation game on $(a,b,\epsilon)$. We show that
$d^W_{n+1}(a,b) \le \epsilon$.
First, for every $A\in\CN$, $|A^\CI(a)-A^\CI(b)| \le \epsilon$ by the
winning condition.
Second, suppose $D$ chooses $\mu\in\cpl(r_a,r_b)$ and 
$\epsilon'\colon \DCI\times\DCI 
\rightarrow[0,1]$ in the first turn.
By assumption, $D$ wins the $n$-round bisimulation game on
$(a',b',\epsilon'(a',b'))$ for every $a',b'\in\DCI$, so
$d^W_n = d^G_n\le\epsilon'$ by induction, and thus $\ev_\mu (d^W_n) 
\le 
\ev_\mu (\epsilon') \le \epsilon$.

``$\le$'': Let $d^W_{n+1}(a,b) < \epsilon$. It suffices to give a
winning strategy for $D$ in the $(n+1)$-round bisimulation game on
$(a,b,\epsilon)$ (implying~$d^G_{n+1}(a,b)\le\epsilon$).  The
winning condition in the initial configuration follows immediately
from the assumption.  Also by the assumption, there exists
$\mu\in\cpl(r_a,r_b)$ such that $\ev_\mu (d^W_n) < \epsilon$.
As $r_a$ and $r_b$ are discrete, the set
\[ R := \{(a',b') \mid r_a(a') > 0 \text{ and } r_b(b') > 0 \}
\]
is countable; so we can write $R = \{(a_1,b_1),(a_2,b_2),\dots\}$.
Now put $\delta = \epsilon - \ev_\mu (d^W_n)$ and define
\[\epsilon'(a_i,b_i) = d^W_n(a_i,b_i) + 2^{-i}\delta \] for
$(a_i,b_i)\in R$ and $\epsilon'(a',b') = 0$ for $(a',b') \notin
R$. Then
\[ \ev_\mu (\epsilon') \le \ev_\mu (d^W_n) + \delta = \epsilon, \]
so playing $\mu$ and $\epsilon'$ constitutes a legal move for $D$.
Now, since $\mu\in\cpl(r_a,r_b)$, $\mu(a',b') = 0$ for all
$(a',b')\notin R$. This means that $S$ must pick some $(a_i,b_i) \in 
R$. Then
\[ d^G_n(a_i,b_i) = d^W_n(a_i,b_i) < \epsilon'(a_i,b_i), \]
so $D$ wins the $n$-round game on $(a_i,b_i,\epsilon'(a_i,b_i))$.

\subsubsection{Proof of Lemma~\ref{lem:diamond-nonexp}.}

Let $\supnorm{f-g}\le\epsilon$; we have to show
$\supnorm{\Diamond f - \Diamond g}\le\epsilon$. So let $a\in
\DCI$; then
\[
|(\Diamond f)(a) - (\Diamond g)(a)|
= \intsuc{a}{f-g}
\le \intsuc{a}{\epsilon} \le \epsilon,
\]
as required.

\subsubsection{Proof of Theorem~\ref{thm:modal-approx}.}

We proceed by simultaneous induction on $n$.

In the base case $n=0$, all the behavioural distances are the zero 
pseudometric: $d^G_0 = 0$ because by the rules of the game each 
$0$-round game is an immediate win for $D$; $d^W_0 = d^K_0 = 0$ by 
definition; and $d^L_0 = 0$ because each rank-$0$ concept is a 
propositional combination of truth constants and therefore constant.
Total boundedness follows directly from the fact that under the zero 
pseudometric every $\epsilon$-ball is the entire space, regardless of 
$\epsilon$. Finally, the density claim follows because non-expansive 
maps under the zero pseudometric are just constants in $[0,1]$ and 
the syntax of $\ALCP$ includes truth constants $q\in\RatI$.

For the inductive step, let $\CI$ be a model and $n > 0$, and assume
as the inductive hypothesis that all claims in
Theorem~\ref{thm:modal-approx} hold for all $n' < n$.  We begin with
Item~\ref{item:metrics-equal}:
\begin{itemize}[wide]
  \item $d^G_n = d^W_n$ is Lemma~\ref{lem:metrics-equal-GW}.
  \item $d^W_n = d^K_n$ follows by Kantorovich-Rubinstein duality
  (Lemma~\ref{lem:kr-duality}), since every totally bounded
  pseudometric space is separable.
  \item $d^K_n = d^L_n$:
  Let $a,b\in\DCI$ and consider the map
  \[
  G\colon\nonexpI{\DCI,d_{n-1}}\to[0,1], \quad
  f \mapsto |(\Diamond f)(a) - (\Diamond f)(b)|,
  \]
  Then~$G$ is a continuous function because all of its constituents 
  are continuous (in particular, $\Diamond$ is continuous by
  Lemma~\ref{lem:diamond-nonexp}).
  
  By the induction hypothesis, and because density is preserved by
  continuous maps, $G[\modf{n-1}]$ is a dense subset of 
  $G[\nonexpI{\DCI,d_{n-1}}]$. Thus,
  \begin{align*}
  d^K_n & (a,b) \\
  & = \bigvee_{A\in\CN} |A^\CI(a)-A^\CI(b)| \vee \bigvee 
  G[\nonexpI{\DCI,d_{n-1}}] \\
  & = \bigvee_{A\in\CN} |A^\CI(a)-A^\CI(b)| \vee \bigvee 
  G[\modf{n-1}] \\
  & = \bigvee_{A\in\CN} |A^\CI(a)-A^\CI(b)| \vee
  \bigvee_{\mathclap{\rk C\le n-1}} |(\Diamond C)^\CI(a)-(\Diamond 
  C)^\CI(b)| 
  \\
  & = \bigvee_{\rk C\le n} |C^\CI(a)-C^\CI(b)| = d^L_n(a,b).
  \end{align*}
  To prove the penultimate step, we first note that ``$\le$'' follows
  immediately. To see ``$\ge$'', we proceed by induction over the
  propositional combinations of atomic concepts $A\in\CN$ and 
  concepts $\Diamond C$, where $C\in\modf{n-1}$, using that for any 
  concepts $C,D$ and $q\in\RatI$:  
  \begin{align*}
  |(C\ominus q)^\CI(a)-(C\ominus q)^\CI(b)| & \le |C^\CI(a)-C^\CI(b)| 
  \\
  |(\neg C)^\CI(a)-(\neg C)^\CI(b)| & = |C^\CI(a)-C^\CI(b)| \\
  |(C\land D)^\CI(a)-(C\land D)^\CI(b)| & \\
  & \hspace{-2.3cm} \le \max(|C^\CI(a)-C^\CI(b)|, 
  |D^\CI(a)-D^\CI(b)|).
  \end{align*}
\end{itemize}
  
  \noindent \emph{Item~\ref{item:tot-bounded}}:
  We make use of the following version of the Arzel\`a-Ascoli 
  theorem~\cite{WildEA18}
  where function spaces are restricted to non-expansive functions
  instead of the more general continuous functions, but the 
  underlying spaces are only required to be totally bounded instead 
  of compact:
  
  \begin{lem}[Arzel\`a-Ascoli for totally bounded
    spaces]\label{lem:arzela-ascoli}
    Let $(X,d)$ be a totally bounded pseudometric space. Then
    the space $\nonexpI{X,d}$, equipped with the supremum
    pseudometric, is totally bounded.
  \end{lem}
  
  By Lemma~\ref{lem:arzela-ascoli}, applied to the inductive 
  hypothesis, we know that the space
  $\nonexpI{\DCI,d_{n-1}}$ is totally bounded wrt.~the supremum
  pseudometric.

Let $\epsilon > 0$. As $\modf{n-1}$ is dense in 
$\nonexpI{\DCI,d_{n-1}}$, there exist
finitely many $C_1,\dots,C_m\in\modf{n-1}$ such that
\[ \bigcup_{i=1}^m \ball{}{\frac{\epsilon}{8}}{C_i} =
\nonexpI{\DCI,d_{n-1}} \]

From these concepts, together with the atomic concepts 
$A_1,\dots,A_k$, we can construct the map
\begin{align*}
I\colon \DCI & \to [0,1]^{k+m} \\
a & \mapsto
(A_1^\CI(a),\dots,A_k^\CI(a),
(\Diamond C_1)^\CI(a),\dots,(\Diamond C_m)^\CI(a)).
\end{align*}

Note that we assume here that the set of atomic concepts is a finite 
set $\CN = \{A_1,\dots,A_k\}$. This is without loss of generality for 
the modal characterization theorem, because every formula of $\FOLP$ 
can only contain finitely many propositional atoms, so $\CN$ can be 
restricted to just those atoms.

It turns out that $I$ is an $\frac{\epsilon}{4}$-isometry, that is
\[ |d_n(a,b) - \supnorm{I(a) - I(b)}| \le \tfrac{\epsilon}{4} \]
for all $a,b\in \DCI$.
Thus, by the triangle inequality, we can take 
preimages to turn a finite $\frac{\epsilon}{4}$-cover of 
$[0,1]^{k+m}$ (a compact, hence totally bounded space) into a finite 
$\epsilon$-cover of $(\DCI,d_n)$.
  
  \emph{Item~\ref{item:modal-approx}}: We make use of the following
  Stone-Weierstra\ss{} theorem ~\cite{WildEA18} (again in a version 
  for totally bounded spaces and non-expansive maps):
  
  \begin{lem}[Stone-Weierstra\ss{} for totally bounded spaces]
    \label{lem:stone-weierstrass}
    Let $(X,d)$ be a totally bounded pseudometric space, and let $L$ 
    be
    a subset of $\nonexpI{X,d}$ such that
    $f_1,f_2 \in L$ implies $\min(f_1,f_2),\max(f_1,f_2) \in L$.
    Then $L$ is dense in $\nonexpI{X,d}$ if each $f\in \nonexpI{X,d}$ 
    can be approximated at each pair of points by functions in~$L$; 
    that is for all $\epsilon>0$ and all $x_1,x_2\in X$ there exists 
    $g\in L$ such that
    \[ \max(|f(x_1)-g(x_1)|,|f(x_2)-g(x_2)|) \le\epsilon. \]
  \end{lem}
  
  We apply Lemma~\ref{lem:stone-weierstrass} to $(\DCI,d_n)$ with
  $L := \modf{n}$. Clearly $L$ is closed under $\min$ and $\max$ so, 
  to finish the proof, it suffices to give, for each $\epsilon>0$, 
  each non-expansive map $f\in\nonexpI{\DCI,d_n}$ and each pair of
  states $a,b\in\DCI$ a concept $C\in\modf{n}$ such that
  \begin{equation*}
  \max(|f(a)-C^\CI(a)|, |f(b)-C^\CI(b)|) \le \epsilon.
  \end{equation*}
  To construct such a~$C$, we note that $|f(a)-f(b)|\le d^L_n(a,b)$
  (by non-expansiveness), so there exists some $D\in\modf{n}$ such
  that $|D^\CI(a)-D^\CI(b)|\ge |f(a)-f(b)|-\epsilon$. From~$D$, we can
  construct~$C$ using truncated subtraction~$\ominus$.
  
\subsubsection{Proof of Lemma~\ref{lem:bounded-morphism-game}.}

We show that $D$ wins the bisimulation game for $(a_0,f(a_0),0)$ by
maintaining the invariant that the current configuration is of the
form $(a,b,0)$ with $b = f(a)$,
which 
ensures that the winning condition always holds. It remains to show
that $D$ can maintain the invariant.

In each round, $D$ begins by picking $\mu(a',b') = r_a(a')$ if
$b'=f(a')$ and $0$ otherwise, and $\epsilon' = 0$.
We can see that $\mu\in\cpl(r_a,r_b)$, because 
\begin{equation*}
\textstyle \sum_{b'\in \DCJ}\mu(a',b') = r_a(a')
\end{equation*}
and
\begin{equation*}
\textstyle\sum_{a'\in \DCI}\mu(a',b') = \sum_{f(a')=b'} r_a(a') = 
r_b(b')
\end{equation*}
for all $a'\in \DCI$ and $b'\in \DCJ$. Also, clearly 
$\ev_\mu(\epsilon') = 0$. Now any choice by $S$ leads to another 
configuration $(a',b',0)$ with $b'=f(a')$.

\subsubsection{Proof of Lemma~\ref{lem:uniform-approx}.}

Let $\CI$ be a model, and let $h\colon\CI\to\CF$ be the unique
morphism. Let $a\in\DCI$. Then $d^G(a,h(a)) = 0$ by
Lemma~\ref{lem:bounded-morphism-game}, and thus
$\phi^\CI(a) = \phi^\CF(h(a))$ and $\psi^\CI(a) = \psi^\CF(h(a))$ by
bisimulation invariance. So
\begin{align*}
\supnorm{\phi-\psi}^\CI & =
\textstyle\bigvee_{a\in \DCI} |\phi^\CI(a)-\psi^\CI(a)|\\& =
\textstyle\bigvee_{a\in \DCI} |\phi^\CF(h(a))-\psi^\CF(h(a))| \\&\le
\supnorm{\phi-\psi}^\CF. 
\end{align*}

\subsubsection{Proof of Lemma~\ref{lem:nbhood-bisim}.}

Player~$D$ wins by maintaining the invariant that whenever $i$ rounds
have been played, the current configuration is of the form
$(a_i,a_i,0)$ for some $a_i\in \DCI$ with $D(a,a_i)\le i$.  For 
$i<k$, no configuration of this kind can be winning for $S$, because 
the two states in this configuration represent the same state in 
different models (recall that the winning conditions are not checked 
after the last round has been played).

It remains to give a strategy for $D$ that maintains the invariant. It
clearly holds at the start of the game, with $a_0 = a$. When the
$(i+1)$-th round is played, $D$ can pick
$\mu\in\cpl(r_{a_i},r_{a_i})$ and $\epsilon'\colon
\DCI\times\nbhood{k}{a} \to [0,1]$ as follows: \begin{align*} 
\mu(a',a'')
& = \begin{cases} \pi_{a_i}(a'), & \text{ if } a' = a'', \\ 0, &
\text{ otherwise}, \end{cases} \\ \epsilon'(a',a'') & = 0.
\end{align*} Clearly, $\ev_\mu(\epsilon') = 0$, so this is a legal
move. Now the new configuration chosen by $S$ necessarily satisfies
the invariant.

\subsubsection{Proof of Lemma~\ref{lem:ef-inv-fol}.}

We proceed by induction over formulae.
\begin{itemize}
  \item The cases $A(x_i)$ and $x_i=x_j$ (with $A\in\CN$) follow
  immediately from the fact that the initial configuration is a
  partial isomorphism.
  \item The Boolean cases ($q, \phi\ominus q, \neg\phi,
  \phi\land\psi$) follow directly by the inductive hypothesis.
  \item $\exists x.\,\phi$:
  Let $(\bar a,\bar b)$ be the current configuration.
  Let $\delta>0$, let $a$ be such that
  \begin{equation*}
  (\exists x.\,\phi)(\bar a) - \phi(\bar a a)< \delta,
  \end{equation*}
  and let $b$ be the winning answer for $D$ in reply to $S$
  choosing $a$. By induction, $\phi(\bar a a) = \phi(\bar b b)$,
  so
  \begin{equation*}
  (\exists x.\,\phi)(\bar b)\ge\phi(\bar b b) = \phi(\bar a a)
  >(\exists x.\,\phi)(\bar a)-\delta.
  \end{equation*}
  Because $\delta>0$ was arbitrary, it follows that 
  $(\exists x.\,\phi)(\bar b)\ge(\exists x.\,\phi)(\bar a)$.
  We can symmetrically show that
  $(\exists x.\,\phi)(\bar a)\ge(\exists x.\,\phi)(\bar b)$, which
  proves this case.
  \item $\diabind{x_i}{x_{m+1}}{\phi}$: Let $(\bar a,\bar b)$ be the
  current configuration.  Suppose that $S$ picks the index $i$ and
  the fuzzy subset
  \begin{equation*}
  \phi_A\colon \DCI\to [0,1], \quad a \mapsto \phi^\CI(\bar a a)
  \end{equation*}
  and $D$'s winning reply is $\psi_B\colon \DCJ\to [0,1]$. We show
  that on the support of $r_{b_i}$, $\psi_B$ must be equal to
  \begin{equation*}
  \phi_B\colon \DCJ\to [0,1], \quad b \mapsto \phi^\CJ(\bar b b).
  \end{equation*}
  Suppose there exists some $b\in \DCJ$ with $r(b_i,b)>0$ and
  $\phi_B(b) \neq \psi_B(b)$. Then $D$ has a winning reply
  $a\in \DCI$ in case $S$ picks this $b$, which means, by the rules
  of the game, that $r(a_i,a)>0$ and $\phi_A(a) =
  \psi_B(b)$.
  However, it is also true that $\phi_A(a) = \phi_B(b)$, by the
  inductive hypothesis.  This is a contradiction.
  
  Now, because $\psi_B$ was a winning reply, we obtain
  \begin{align*}
  (\diabind{x_i}{x_{m+1}}{\phi})(\bar a) & =
  \intsuc{a_i}{\phi_A}\\& =
  \intsuc{b_i}{\psi_B}\\& =
  \intsuc{b_i}{\phi_B} \\&=
  (\diabind{x_i}{x_{m+1}}{\phi})(\bar b). 
  \end{align*}
\end{itemize}

\subsubsection{Proof of Lemma~\ref{lem:bisim-inv-local}.}

Let $a$ be a state in a model~$\CI$. We
need to show $\phi^\CI(a) = \phi^{\CI^k_a}(a)$.  Let $\CJ$ be a new
model that extends $\CI$ by adding $n$ disjoint copies of both $\CI$
and $\CI^k_a$. Let $\CK$ be the model that extends $\CI^k_a$
likewise.  We finish the proof by showing that
\begin{equation*}
\phi^\CI(a) = \phi^\CJ(a) = \phi^\CK(a) = \phi^{\CI^k_a}(a).
\end{equation*}
The first and third equality follow by bisimulation invariance of
$\phi$ (Lemma~\ref{lem:bisim-inv-disjoint}).
The second equality follows by Ehrenfeucht-Fra\"iss\'e invariance
(Lemma~\ref{lem:ef-inv-fol}) once we show that $D$ has a winning
strategy in the $n$-round Ehrenfeucht-Fra\"iss\'e game for $\CJ,a$
and $\CK,a$.

Such a winning strategy can be described as follows: For
$\bar a = (a_1,\dots,a_n)$, put
$\nbhood{k}{\bar a} = \bigcup_{i\le n}\nbhood{k}{a_i}$. Then $D$
maintains the invariant that, if the configuration reached after $i$
rounds is $(\bar b,\bar c)$, then there exists an isomorphism $f_i$
between $\nbhood{k_i}{\bar b}$ and $\nbhood{k_i}{\bar c}$ that maps
each $b_j$ to the corresponding $c_j$, where $k_i = 3^{n-i}$.

The invariant holds at the start of the game, because the
neighbourhoods on both sides are just $\nbhood{k}{a}$.
Similarly, whenever the invariant holds, the current configuration
is a partial isomorphism by restriction of the given isomorphism to
the two vectors of the configuration.

Now we consider what happens during the rounds. Suppose that $i$
rounds have been played, and the current configuration is $(\bar
b,\bar c)$. If $S$ decides to play a standard round, playing some
$b\in \DCJ$, then there are two cases:
\begin{itemize}
  \item $b\in \nbhood{2k_{i+1}}{\bar b}$:
  In this case, the radius-$k_{i+1}$ neighbourhood
  $\nbhood{k_{i+1}}{b}$ of $b$ is fully contained in the domain
  $\nbhood{k_i}{\bar b}$ of $f_i$ -- this follows by the triangle
  inequality, as $2k_{i+1} + k_{i+1} = 3k_{i+1} = k_i$. Now $D$
  can just reply with $c := f_i(b)$, and an isomorphism $f_{i+1}$
  between $\nbhood{k_{i+1}}{\bar bb}$ and $\nbhood{k_{i+1}}{\bar
    cc}$ is formed by restricting the domain and codomain of $f_i$
  appropriately.
  \item $b\notin \nbhood{2k_{i+1}}{\bar b}$:
  In this case, the radius-$k_{i+1}$ neighbourhoods
  $\nbhood{k_{i+1}}{b}$ of $b$ and $\nbhood{k_{i+1}}{\bar b}$ of
  $\bar b$ do not intersect -- this too follows from the triangle
  inequality.
  Now $D$ can pick a fresh copy of $\CI$ or $\CI^k_a$ in $\CK$
  (depending on which kind of copy $b$ lies in); her reply $c$ is
  then just $b$ in that copy. Here, a fresh copy is one that was
  never visited on any of the previous rounds. By construction of
  $\CJ$ and $\CK$, such a copy is always available.
  This means that we now have two isomorphisms, one between
  $\nbhood{k_{i+1}}{\bar b}$ and $\nbhood{k_{i+1}}{\bar c}$ (by
  restriction of $f_i$), and one between
  $\nbhood{k_{i+1}}{b}$ and $\nbhood{k_{i+1}}{c}$ (by isomorphism
  of the respective copies of $\CI$ or $\CI^k_a$). Because these
  isomorphisms have disjoint domains and codomains, we can combine
  them to form the desired isomorphism $f_{i+1}$.
\end{itemize}
If $S$ plays a standard round with some $c\in \Delta^\CK$ instead, the
same argument applies.

Finally, if $S$ starts a probabilistic round by picking an index
$0\le j\le i$ and playing some $\phi_B\colon \DCJ\to [0,1]$, then we
first note that, by the rules of the game, the support of $\phi_B$
must be contained in $\nbhood{1}{\bar b}$, which in turn must be
contained in the domain of $f_i$. This means that $D$ can construct
$\phi_C\colon \Delta^\CK\to [0,1]$ by mapping along $f_i$, 
i.e.~$\phi_C(c) =
\phi_B(f_i^{-1}(c))$ for all successors $c$ of $c_j$, and
$\phi_C(c) = 0$ otherwise. Now, whichever $b$ or $c$ is picked by
$S$, $D$ can just reply with $c:=f_i(b)$ or $b:=f_i^{-1}(c)$ and
$f_{i+1}$ is formed as in the first case of a standard round.
Again, the same argument applies if $S$ picks a fuzzy subset
$\phi_C$ on the other side.

\subsubsection{Proof of Lemma~\ref{lem:bisim-unravel}.}

$D$ wins by maintaining the invariant that the configuration of the
game is of the form $(\bar a,\last(\bar a),0)$ for some $\bar a\in 
(\DCI)^+$.
To do so, she can put $\mu(\bar aa,a) = \pi_{\bar a}(\bar aa) =
\pi_{\last(\bar a)}(a)$ for all $a\in (\DCI)^+$, all other values of
$\mu$ are $0$, and $\epsilon' = 0$. Then any move by $S$ leads to a
configuration where the invariant holds.

\subsubsection{Proof of Lemma~\ref{lem:local-k-bisim-inv}.}

Let $\CI$ and $\CJ$ be two models and
let $a\in \DCI$ and $b\in \DCJ$ be two states such that
$d_k^G(a,b)<\epsilon$. It is enough to show that
$|\phi^\CI(a)-\phi^\CJ(b)|\le\epsilon$.

We denote by $a'$ and $a''$ the copies of $a$ in $\CI^\ast$ and
$(\CI^\ast)^k_a$, respectively. Similarly, $b'$ and $b''$ denote the
copies of $b$ in $\CJ^\ast$ and $(\CJ^\ast)^k_b$.
By Lemma~\ref{lem:bisim-unravel}, $D$ wins the
$0$-bisimulation-game for $\CI,a$ and $\CI^\ast,a'$ (similarly for
$\CJ$) and by Lemma~\ref{lem:nbhood-bisim}, she also wins the
$k$-round $0$-bisimulation game for $\CI^\ast,a'$ and
$(\CI^\ast)^k_a,a''$ (similarly for $\CJ$).
Because behavioural distance $d^G_k$ is a pseudometric, this
means that
\begin{align*}
&d^G_k(a'',b'')\\& \le
d^G_k(a'',a')+d^G_k(a',a)\\&\qquad +d^G_k(a,b)+d^G_k(b,b')+d^G_k(b',b'')\\
& =d^G_k(a,b)<\epsilon,
\end{align*}
so $D$ has a winning strategy in the $k$-round
$\epsilon$-bisimulation game for $(\CI^\ast)^k_a,a''$ and
$(\CJ^\ast)^k_b,b''$.

In both $(\CI^\ast)^k_a,a''$ and $(\CJ^\ast)^k_b,b''$, the reachable
states form a tree of depth at most $k$. This implies that, after
$i$ rounds of the game, the two states on either side of the current
configuration are nodes at distance $i$ from the root of their
respective tree. Thus, whenever $k$ rounds have been played in the
game, $S$ does not have a legal move in the next round, because at
that point, both nodes in the configuration are necessarily leaves
and thus blocking.
This in turn means that if $D$ can win the $k$-round game, she
also wins the unbounded game, so, by bisimulation invariance of
$\phi$, 
$|\phi^{(\CI^\ast)^k_a}(a'')-\phi^{(\CJ^\ast)^k_b}(b'')|\le\epsilon$.

By locality and bisimulation invariance of $\phi$, and again
Lemma~\ref{lem:bisim-unravel}, we have
$\phi^{(\CI^\ast)^k_a}(a'') = \phi^{\CI^\ast}(a') =
\phi^\CI(a)$
as well as
$\phi^{(\CJ^\ast)^k_b}(b'') = \phi^{\CJ^\ast}(b') =
\phi^\CJ(b)$.
Thus $|\phi^\CI(a)-\phi^\CJ(b)|\le\epsilon$, as claimed.

\end{document}

%% file: defsplain.tex
\theoremstyle{definition}
       \newtheorem{defn}{Definition}[section]

       \newtheorem{rem}[defn]{Remark}
\theoremstyle{plain}
       \newtheorem{thm}[defn]{Theorem}
       
       \newtheorem{lem}[defn]{Lemma}